\documentclass[11pt]{article}
\usepackage[pdftex]{graphicx}
\usepackage[T1]{fontenc}
\usepackage{lmodern}
\usepackage{fullpage}
\usepackage{amsmath,amsthm,amssymb}
\usepackage{verbatim}
\usepackage{mathtools}
\usepackage{algorithm}
\usepackage{algorithmicx}
\usepackage{algpseudocode}
\usepackage{subfig}
\usepackage[toc,page]{appendix}
\usepackage[usenames,dvipsnames]{color}
\usepackage{caption}
\usepackage{thmtools}
\usepackage{thm-restate}
\usepackage{hyperref}
\usepackage{cleveref}
\usepackage{tabularx}

\declaretheorem[name=Theorem,numberwithin=section]{thm}

\declaretheorem[name=Lemma,numberlike=thm]{lemma}

\declaretheorem[name=Definition,numberlike=thm,style=definition]{defn}

\makeatletter
\newcommand\footnoteref[1]{\protected@xdef\@thefnmark{\ref{#1}}\@footnotemark}
\makeatother

\DeclareMathOperator*{\argmax}{arg\,max}

\newcommand{\F}{\phi}

\newcommand{\R}{\mathbb{R}}
\newcommand{\N}{\mathbb{N}}
\newcommand{\bbP}{\mathbb{P}}

\newcommand{\Z}{\mathbb{Z}_{+}}

\newcommand{\simplex}{\Delta^{\bX}}

\newcommand{\probpml}{\bbP}

\newcommand{\bp}{\textbf{p}}

\newcommand{\bff}{\textbf{f}}

\newcommand{\bX}{\mathcal{D}}

\newcommand{\expo}[1]{\exp \left(#1 \right)}
\newcommand{\exps}[1]{\exp (#1 )}

\newcommand{\eqdef}{\stackrel{\mathrm{def}}{=}}
\newcommand{\defeq}{\eqdef}

\newcommand{\bgg}{\textbf{g}}

\newcommand{\bpml}{\bp^{\beta}_{\phis}}
\newcommand{\Phisn}{\Phi_{S}^{n}}

\newcommand{\bset}{\textbf{B}_{\fsub,S,\hat{\bff}}}
\newcommand{\probbp}[1]{\bp\left(#1\right)}
\newcommand{\probbpml}[1]{\bp^{\beta}_{\phis}\left(#1\right)}
\newcommand{\estiS}[1]{\hat{\bff}(#1)}
\newcommand{\Prob}[1]{\bbP\left(#1 \right)}
\newcommand{\xon}{x_1^n}
\newcommand{\xtn}{x_2^n}
\newcommand{\fsub}{\mathrm{F}}

\newcommand{\phis}{\phi_{S}}

\usepackage[utf8]{inputenc} % allow utf-8 input
\usepackage[T1]{fontenc}    % use 8-bit T1 fonts
\usepackage{hyperref}       % hyperlinks
\usepackage{url}            % simple URL typesetting
\usepackage{booktabs}       % professional-quality tables
\usepackage{amsfonts}       % blackboard math symbols
\usepackage{nicefrac}       % compact symbols for 1/2, etc.
\usepackage{microtype}      % microtypography
\usepackage{enumitem}       % customize lists

\title{A General Framework for Symmetric Property Estimation}

\author{%
  Moses Charikar\\
Stanford University\\
\texttt{moses@cs.stanford.edu} \\
\and
Kirankumar Shiragur\\
Stanford University\\
\texttt{shiragur@stanford.edu} \\
\and
Aaron Sidford\\
Stanford University\\
\texttt{sidford@stanford.edu} \\
}

\begin{document}

\maketitle

\begin{abstract}
In this paper we provide a general framework for estimating symmetric properties of distributions from i.i.d. samples. For a broad class of symmetric properties we identify the  {\em easy} region where empirical estimation works and the {\em difficult} region where more complex estimators are required. We show that by approximately computing the profile maximum likelihood (PML) distribution \cite{ADOS16} in this difficult region we obtain a symmetric property estimation framework that is sample complexity optimal for many properties in a broader parameter regime than previous universal estimation approaches based on PML. The resulting algorithms based on these \emph{pseudo PML distributions} are also more practical.
\end{abstract}

\section{Introduction}
Symmetric property estimation is a fundamental and well studied problem in machine learning and statistics. In this problem, we are given $n$ i.i.d samples from an unknown distribution\footnote{Throughout the paper, distribution refers to discrete distribution.} $\bp$ and asked to estimate $\bff(\bp)$, where $\bff$ is a symmetric property (i.e. it does not depend on the labels of the symbols).
Over the past few years, the computational and sample complexities for estimating many symmetric properties have been extensively studied.
Estimators with optimal sample complexities have been obtained for several properties including entropy~\cite{VV11a, WY16, JVHW15}, distance to uniformity~\cite{VV11b, JHW16},  and support~\cite{VV11a, WY15}. 

All aforementioned estimators were property specific and therefore, a natural question is to design a universal estimator. 
In \cite{ADOS16}, the authors showed that the distribution that maximizes the profile likelihood, i.e. the likelihood of the multiset of frequencies of elements in the sample, referred to as \emph{profile maximum likelihood (PML) distribution}, can be used as a universal plug-in estimator. \cite{ADOS16} showed that computing the symmetric property on the PML distribution is sample complexity optimal in estimating support, support coverage, entropy and distance to uniformity within accuracy $\epsilon > \frac{1}{n^{0.2499}}$. Further, this also holds for distributions that approximately optimize the PML objective, where the approximation factor affects the desired accuracy.

Acharya et al. \cite{ADOS16} posed two important and natural open questions. The first was to give an efficient algorithm for finding an approximate PML distribution, which was recently resolved in \cite{CSS19}. 
The second open question is whether PML is sample competitive in all regimes of the accuracy parameter $\epsilon$? In this work, we make progress towards resolving this open question.

First, we show that the PML distribution based plug-in estimator achieves optimal sample complexity for all $\epsilon$ for the problem of estimating support size.
Next, we introduce a variation of the PML distribution that we call the \emph{pseudo PML distribution}.
Using this, we give a general framework for estimating a symmetric property.
For entropy and distance to uniformity, this pseudo PML based framework achieves optimal sample complexity for a broader regime of the accuracy parameter than was known for the vanilla PML distribution.

We provide a general framework that could, in principle be applied to estimate any separable symmetric property $\bff$, meaning $\bff(\bp)$ can be written in the form of $\sum_{x \in \bX}\bff(\bp_x)$. This motivation behind this framework is that for any symmetric property $\bff$ that is separable, the estimate for $\bff(\bp)$ can be split into two parts: $\bff(\bp)=\sum_{x \in B}\bff(\bp_x)+\sum_{x \in G}\bff(\bp_x)$, where $B$ and $G$ are a (property dependent) disjoint partition of the domain $\bX$. 
We refer to $G$ as the good set and $B$ as the bad set. Intuitively, $G$ is the subset of domain elements whose contribution to $\bff(\bp)$ is easy to estimate, i.e a simple estimator such as empirical estimate (with correction bias) works. 
For many symmetric properties, finding an appropriate partition of the domain is often easy.
Many estimators in the literature~\cite{JVHW15, JHW16, WY16} make such a distinction between domain elements.
The more interesting and difficult case is estimating the contribution of the bad set: $\sum_{x \in B}\bff(\bp_x)$.
Much of the work in these estimators is dedicated towards estimating this contribution
using sophisticated techniques such as polynomial approximation.
Our work gives a unified approach to estimating the contribution of the bad set. We propose a PML based estimator for estimating $\sum_{x \in B}\bff(\bp_x)$. 
We show that computing the PML distribution only on the set $B$ is sample competitive for entropy and distance to uniformity for almost all interesting parameter regimes thus (partially) handling the open problem proposed in \cite{ADOS16}. 
Additionally, requiring that the PML distribution be computed on a subset $B \subseteq \bX$ reduces the input size for the PML subroutine and results in practical algorithms (See \Cref{sec:exp}).

To summarize, the main contributions of our work are:
\vspace*{-0.1in}
\begin{itemize}
\item We make progress on an open problem of \cite{ADOS16} on broadening the range of error parameter $\epsilon$ that one can obtain for universal symmetric property estimation via PML.
\item We give a general framework for applying PML to new symmetric properties.
\item As a byproduct of our framework, we obtain more practical algorithms that invoke PML on smaller inputs (See \Cref{sec:exp}).
\end{itemize}

\subsection{Related Work}
For many natural properties, there has been extensive work on designing efficient estimators both with respect to computational time and sample complexity \cite{HJW17, HJM17, AOST14, RVZ17, ZVVKCSLSDM16, WY16a, RRSS07, WY15, OSW16, VV11a, WY16, JVHW15, JHW16, VV11b}.
We define and state the optimal sample complexity for estimating support, entropy and distance to uniformity. For entropy, we also discuss the regime in which the empirical distribution is sample optimal.

\noindent{\bf Entropy:} For any distribution $\bp \in \simplex$, the entropy $H(\bp)\defeq -\sum_{x\in \bX}\bp_{x} \log \bp_{x}$. For $\epsilon \geq  \frac{\log N}{N}$ (the interesting regime), where $N\defeq |\bX|$, the optimal sample complexity for estimating $H(\bp)$ within additive accuracy $\epsilon$ is $O(\frac{N}{\log N}\frac{1}{\epsilon})$~\cite{WY16}. Further if $\epsilon <  \frac{\log N}{N}$, then \cite{WY16} showed that empirical distribution is optimal.

\noindent{\bf Distance to uniformity:} For any distribution $\bp \in \simplex$, the distance to uniformity $\|\bp-u\|_1\defeq \sum_{x\in \bX}|\bp_{x}-\frac{1}{N}|$, where $u$ is the uniform distribution over $\bX$. The optimal sample complexity for estimating $\|\bp-u\|_1$ within additive accuracy $\epsilon$ is $O(\frac{N}{\log N}\frac{1}{\epsilon^2})$~\cite{VV11b, JHW16}.

\newcommand{\lprob}{k}
\noindent{\bf Support:} For any distribution $\bp \in \simplex$, the support of distribution $S(\bp)\defeq |\{x\in \bX~|~\bp_{x}>0 \}|$. Estimating support is difficult in general because we need sufficiently large number of samples to observe elements with small probability values. Suppose for all $x \in \bX$, if $\bp_{x} \in \{0\} \cup [\frac{1}{\lprob},1]$, then \cite{WY15} showed that the optimal sample complexity for estimating support within additive accuracy $\epsilon \lprob$ is $O(\frac{\lprob}{\log \lprob}\log^2 \frac{1}{\epsilon})$.

PML was introduced by Orlitsky et al.~\cite{OSSVZ04} in 2004. 
The connection between PML and universal estimators was first studied in~\cite{ADOS16}. 
As discussed in the introduction, PML based plug-in estimator applies to a restricted regime of error parameter $\epsilon$.
There have been several other approaches for designing universal estimators for symmetric properties. Valiant and Valiant~\cite{VV11a} adopted and rigorously analyzed a linear programming based approach for universal estimators proposed by \cite{ET76} and showed that it is sample complexity optimal in the constant error regime for estimating certain symmetric properties (namely, entropy and support size). Recent work of Han et al.~\cite{HJW18} applied a local moment matching based approach in designing efficient universal symmetric property estimators for a single distribution. \cite{HJW18} achieves the optimal sample complexity in restricted error regimes for estimating the power sum function, support and entropy.

Recently, \cite{YOSW18} gave a different unified approach to property estimation.
They devised an estimator that uses $n$ samples and achieves the performance attained by the empirical estimator with $n \sqrt{\log n}$ samples for a wide class of properties and for all underlying distributions. 
This result is further strengthened to $n \log n$ samples for Shannon entropy and a broad class of other properties including $\ell_1$-distance in \cite{HO19b}.

Independently of our work, authors in \cite{HO19a} propose \emph{truncated PML} that is slightly different but similar in the spirit to our idea of pseudo PML. They use the approach of truncated PML and study its application to symmetric properties such as: entropy, support and coverage; refer \cite{HO19a} for further details.
\subsection{Organization of the Paper}
In \Cref{sec:prelims} we provide basic notation and definitions. 
We present our general framework in \Cref{sec:results} and state all our main results. 
In \Cref{app:universal}, we provide proofs of the main results of our general framework. 
In \Cref{sec:appl}, we use these results to establish the sample complexity of our estimator in the case of entropy (See \Cref{sec:entr}) and distance to uniformity (See \Cref{sec:dtu}). Due to space constraints, many proofs are deferred to the appendix. In  \Cref{sec:exp}, we provide experimental results for estimating entropy using pseudo PML and other state-of-the-art estimators. Here we also demonstrate the practicality of our approach.

\newcommand{\setd}{\textbf{D}}
\newcommand{\eled}{d}
\newcommand{\phik}{\psi}
\newcommand{\phikf}{h}
\newcommand{\Phik}{\Psi}
\newcommand{\Phis}{\Phi_{S}}
\newcommand{\Phikn}{\Psi^{(k,n)}}
\newcommand{\Fk}{\psi}
\newcommand{\sk}{\textbf{S}_{\phik}}
\newcommand{\dfreq}[1]{\mathrm{Freq}\left(#1\right)}

\section{Preliminaries}\label{sec:prelims}
Let $[a]$ denote all integers in the interval $[1,a]$. Let $\simplex \subset [0,1]_{\R}^{\bX}$ be the set of all distributions supported on domain $\bX$ and let $N$ be the size of the domain. Throughout this paper we restrict our attention to discrete distributions and assume that we receive a sequence of $n$ independent samples from an underlying distribution $\bp \in \simplex$. Let $\bX ^n$ be the set of all length $n$ sequences and $y^n \in \bX^n$ be one such sequence with $y^n_{i}$ denoting its $i$th element. The probability of observing sequence $y^n$ is:
$$\bbP(\bp,y^n) \defeq \prod_{x \in \bX}\bp_x^{\bff(y^n,x)}$$
where $\bff(y^n,x)= |\{i\in [n] ~ | ~ y^n_i = x\}|$ is the frequency/multiplicity of symbol $x$ in sequence $y^n$ and $\bp_x$ is the probability of domain element $x\in \bX$. We next formally define profile, PML distribution and approximate PML distribution.

\begin{defn}[Profile]
For a sequence $y^n \in \bX^n$, its \emph{profile} denoted $\phi=\Phi(y^n) \in \Z^{n}$ is $\phi\defeq(\F(j))_{j \in [n]} $ where $\F(j)\defeq|\{x\in \bX |\bff(y^n,x)=j \}|$ is the number of domain elements with frequency ${j}$ in $y^{n}$. We call $n$ the length of profile $\F$ and use $\Phi^n$ denote the set of all profiles of length $n$.
\footnote{The profile does not contain $\F(0)$, the number of unseen domain elements.}
\end{defn}

For any distribution $\bp \in \simplex$, the probability of a profile $\phi \in \Phi^n$ is defined as:
\begin{equation}\label{eqpml1}
\probpml(\bp,\phi)\defeq\sum_{\{y^n \in \bX^n~|~ \Phi (y^n)=\phi \}} \bbP(\bp,y^n)\\
\end{equation}

The distribution that maximizes the probability of a profile $\phi$ is the profile maximum likelihood distribution and we formally define it next.
\begin{defn}[Profile maximum likelihood distribution] For any profile $\phi \in \Phi^{n}$, a \emph{Profile Maximum Likelihood} (PML) distribution $\bp_{pml,\phi} \in \simplex$ is:
 $\bp_{pml,\phi} \in \argmax_{\bp \in \simplex} \probpml(\bp,\phi)$ and $\probpml(\bp_{pml,\phi},\phi)$ is the maximum PML objective value. Further, a distribution $\bp^{\beta}_{pml,\phi} \in \simplex$ is a $\beta$-\emph{approximate PML} distribution if $\probpml(\bp^{\beta}_{pml,\phi},\phi)\geq \beta \cdot \probpml(\bp_{pml,\phi},\phi)$.
\end{defn}

We next provide formal definitions for separable symmetric property and an estimator.
\begin{defn}[Separable Symmetric Property] A symmetric property $\bff:\simplex \rightarrow \R$ is separable if for any $\bp \in \simplex$, $f(\bp)\defeq\sum_{x\in \bX}\bgg(\bp_{x})$, for some function $\bgg:\R \rightarrow \R$. Further for any subset $S\subset \bX$, we define $f_{S}(\bp)\defeq\sum_{x\in S}\bgg(\bp_{x})$.
	\end{defn}

\begin{defn}
	A property estimator is a function $\hat{\bff}:\bX^n \rightarrow \R$, that takes as input $n$ samples and returns the estimated property value. The sample complexity of $\hat{\bff}$ for estimating a symmetric property $\bff(\bp)$ is the number of samples needed to estimate $\bff$ up to accuracy $\epsilon$ and with constant probability. The optimal sample complexity of a property $\bff$ is the minimum number of samples of any estimator.
	\end{defn} 
\newcommand{\co}{c_{1}}
\newcommand{\ct}{2c_{1}}
\newcommand{\fsubp}{\mathrm{F'}}
\newcommand{\consto}{2c_{1}}
\newcommand{\constt}{c_{5}}
\newcommand{\gset}{\mathrm{G}}
\newcommand{\bpmlp}{\bp_{\phi_{S'}}^{\beta}}
\newcommand{\phisp}{\phi_{S'}}
\newcommand{\rvS}{\textbf{S}}

\section{Main Results}\label{sec:results}
As discussed in the introduction, one of our motivations was to provide a better analysis for the PML distribution based plug-in estimator. In this direction, we first show that the PML distribution is sample complexity optimal in estimating support in all parameter regimes. Estimating support is difficult in general and all previous works make the assumption that the minimum non-zero probability value of the distribution is at least $\frac{1}{k}$. In our next result, we show that the PML distribution under this constraint is sample complexity optimal for estimating support.
\begin{thm}\label{thm:supp}
	The PML distribution \footnote{\label{suppfoot} Under the constraint that its minimum non-zero probability value is at least $\frac{1}{k}$. This assumption is also necessary for the results in \cite{ADOS16} to hold.} based plug-in estimator is sample complexity optimal in estimating support for all regimes of error parameter $\epsilon$.
\end{thm}
For support, we show that an approximate PML distribution is sample complexity optimal as well.
\begin{thm}\label{thm:approxsupp}
	For any constant $\alpha>0$, an $\exps{-\epsilon^2 n^{1-\alpha}}$-approximate PML distribution \footnoteref{suppfoot} based plug-in estimator is sample complexity optimal in estimating support for all regimes of error $\epsilon$.
\end{thm}
We defer the proof of both these theorems to \Cref{sec:support}.

For entropy and distance to uniformity, we study a variation of the PML distribution we call the pseudo PML distribution and present a general framework for symmetric property estimation based on this. We show that this pseudo PML based general approach gives an estimator that is sample complexity optimal for estimating entropy and distance to uniformity in broader parameter regimes. To motivate and understand this general framework we first define new generalizations of the profile, PML and approximate PML distributions.
\newcommand{\Fsj}{\phis(j)}
\begin{defn}[$S$-pseudo Profile]
	For any sequence $y^{n} \in \bX^{n}$ and $S \subseteq \bX$, its $S$-\emph{pseudo} profile denoted $\phis=\Phis(y^n)$ is $\phi\defeq(\Fsj)_{j\in[n]} $ where $\Fsj\defeq|\{x\in S ~|~\bff(y^n,x)=j \}|$ is the number of domain elements in $S$ with frequency ${j}$ in $y^{n}$. We call $n$ the length of $\phis$ as it represents the length of the sequence $y^n$ from which this pseudo profile was constructed. Let $\Phisn$ denote the set of all $S$-pseudo profiles of length $n$.
\end{defn}

For any distribution $\bp \in \simplex$, the probability of a $S$-pseudo profile $\phis \in \Phisn$ is defined as:
\begin{equation}\label{eqpml1}
\probpml(\bp,\phis)\defeq\sum_{\{y^n \in \bX^n~|~ \Phis(y^n)=\phis \}} \bbP(\bp,y^n)\\
\end{equation}
We next define the $S$-pseudo PML and $(\beta,S)$-approximate pseudo PML distributions that are analogous to the PML and approximate PML distributions.
\begin{defn}[$S$-pseudo PML distribution]
	For any $S$-pseudo profile $\phis \in \Phisn$, a distribution $\bp_{\phis} \in \simplex$ is a $S$-\emph{pseudo PML} distribution if $\bp_{\phis} \in \argmax_{\bp \in \simplex} \probpml(\bp,\phis)$.
\end{defn}

\begin{defn}[$(\beta,S)$-approximate pseudo PML distribution]
	For any profile $\phis \in \Phisn$, a distribution $\bpml \in \simplex$ is a $(\beta,S)$-\emph{approximate pseudo PML} distribution if $\probpml(\bpml,\phis)\geq \beta \cdot \probpml(\bpml,\phis)$.
\end{defn}
For notational convenience, we also define the following function.
\begin{defn}
	For any subset $S \subseteq \bX$, the function $\mathrm{Freq}:\Phisn \rightarrow 2^{\Z}$ takes input a $S$-psuedo profile and returns the set with all distinct frequencies in $\phis$.
\end{defn}

Using the definitions above, we next give an interesting generalization of Theorem 3 in \cite{ADOS16}.
\begin{thm}\label{thmbeta}
	For a symmetric property $\bff$ and $S \subseteq \bX$, suppose there is an estimator $\hat{\bff}:\Phisn \rightarrow \R$, such that for any $\bp$ and $\phis \sim \bp$ the following holds, $$\Prob{|\bff_{{S}}(\bp)-\hat{\bff}(\phis)|\geq \epsilon} \leq \delta~,$$ then for any $\fsub \in 2^{\Z}$, a $(\beta,S)$-approximate pseudo PML distribution $\bpml$ satisfies: $$\Prob{|\bff_{{S}}(\bp)-\bff_{{S}}(\bpml)| \geq 2 \epsilon} \leq \frac{\delta n^{|\fsub|}}{\beta} \Prob{\dfreq{\phis} \subseteq \fsub} +\Prob{\dfreq{\phis} \not\subseteq \fsub}~.$$
\end{thm}

Note that in the theorem above, the error probability with respect to a pseudo PML distribution based estimator has dependency on$\frac{\delta n^{|\fsub|}}{\beta}$ and $\Prob{\dfreq{\phis} \not\subseteq \fsub}$. However Theorem 3 in \cite{ADOS16} has error probability $\frac{\delta e^{\sqrt{n}}}{\beta}$. This is the bottleneck in showing that PML works for all parameter regimes and the place where pseudo PML wins over the vanilla PML based estimator, getting non-trivial results for entropy and distance to uniformity. We next state our general framework for estimating symmetric properties.
We use the idea of sample splitting which is now standard in the literature \cite{WY16,JVHW15,JHW16,CL11,Nem03}.
\begin{algorithm}[H]
	\caption{General Framework for Symmetric Property Estimation}\label{algpml}
	\begin{algorithmic}[1]
		\Procedure{Property estimation}{$x^{2n},\bff, \fsub$}
		\State Let $x^{2n}=(\xon,\xtn)$, where $\xon$ and $\xtn$ represent first and last $n$ samples of $x^{2n}$ respectively.
		\State Define $S\defeq \{y \in \bX~|~f(\xon,y) \in \fsub \}$.
		\State Construct profile $\phis$, where $\Fsj\defeq|\{y\in S ~|~\bff(\xtn,y)=j \}|$.
		\State Find a $(\beta,S)$-approximate pseudo PML distribution $\bpml$ and empirical distribution $\hat{\bp}$ on $\xtn$.
		\State \textbf{return} $\bff_{S}(\bpml)+\bff_{\bar{S}}(\hat{\bp}) + \text{correction bias with respect to }\bff_{\bar{S}}(\hat{\bp})$.
		\EndProcedure
	\end{algorithmic}
\end{algorithm}
In the above general framework, the choice of $\fsub$ depends on the symmetric property of interest. Later, in the case of entropy and distance to uniformity, we will choose $\fsub$ to be the region where the empirical estimate fails; it is also the region that is difficult to estimate. One of the important properties of the above general framework is that $\bff_{S}(\bpml)$ (recall $\bpml$ is a $(\beta,S)$-approximate pseudo PML distribution and $\bff_{S}(\bpml)$ is the property value of distribution $\bpml$ on subset of domain elements $S \subseteq \bX$) is close to $\bff_{S}(\bp)$ with high probability. Below we state this result formally.
\begin{thm}\label{thm:main}
	For any symmetric property $\bff$, let $\gset \subseteq {\bX}$ and $\fsub,\fsubp \in 2^{\Z}$. If for all $S' \in 2^\gset$, there exists an estimator $\hat{\bff}:\Phi_{S'}^{n} \rightarrow \R$, such that for any $\bp$ and $\phisp \sim \bp$ satisfies, 
	\begin{equation}\label{eq:condio}
	\Prob{|\bff_{{S'}}(\bp)-\hat{\bff}(\phisp)|\geq2 \epsilon} \leq \delta \text{ and }\Prob{\dfreq{\phisp} \not\subseteq \fsub'} \leq \gamma~.
	\end{equation}
	Then for any sequence $x^{2n}=(\xon,\xtn)$, $$\Prob{|\bff_{S}(\bp)-\bff_{S}(\bpml)|>4\epsilon}\leq \frac{\delta n^{|F'|}}{\beta}+\gamma + \Prob{S \notin 2^\gset}~,$$ where $S$ is a random set $\rvS\defeq \{y \in \bX~|~f(\xon,y) \in \fsub \}$ and $\phis\defeq \Phis(\xtn)$.
\end{thm}

Using the theorem above, we already have a good estimate for $\bff_{S}(\bp)$ for appropriately chosen frequency subsets $\fsub, \fsubp$ and $\gset \subseteq \bX$. Further, we choose these subsets $\fsub, \fsubp$ and $\gset$ carefully so that the empirical estimate $\bff_{\bar{S}}(\hat{p})$ plus the correction bias with respect to $\bff_{\bar{S}}$ is close to $\bff_{\bar{S}}(\bp)$. Combining these together, we get the following results for entropy and distance to uniformity. 
\begin{thm}\label{thm:entr}
	If error parameter $\epsilon>\Omega\left(\frac{\log N}{N^{1-\alpha}}\right)$ for any constant $\alpha>0$,
	then for estimating entropy, the estimator \ref{algpml} for $\beta=n^{-\log n}$ is sample complexity optimal.
\end{thm}
For entropy, we already know from \cite{WY16} that the empirical distribution is sample complexity optimal if $\epsilon <c \frac{\log N}{N}$ for some constant $c>0$. Therefore the interesting regime for entropy estimation is when $\epsilon>\Omega\left(\frac{\log N}{N}\right)$ and our estimator works for almost all such $\epsilon$.
\begin{thm}\label{thm:dtu}
	Let $\alpha>0$ and error parameter  $\epsilon>\Omega\left(\frac{1}{N^{1-8\alpha}}\right)$, then for estimating distance from uniformity, the estimator \ref{algpml} for $\beta=n^{-\sqrt{\frac{n \log n}{N}}}$ is sample complexity optimal.
\end{thm}
Note that the estimator in \cite{JHW17} also requires that the error parameter $\epsilon \geq \frac{1}{N^C}$, where $C>0$ is some constant.
\section{Analysis of General Framework for Symmetric Property Estimation}\label{app:universal}
Here we provide proofs of the main results for our general framework (Theorem`\ref{thmbeta} and \ref{thm:main}). These results weakly depend on the property and generalize results in \cite{ADOS16}. The PML based estimator in \cite{ADOS16} is sample competitive only for a restricted error parameter regime and this stems from the large number of possible profiles of length $n$. Our next lemma will be useful to address this issue and later we show how to use this result to prove Theorems \ref{thmbeta} and \ref{thm:main}.

\begin{lemma}\label{lem:bcard}
	 For any subset $S \subseteq \bX$ and $\fsub \in 2^{\Z}$, if set $B$ is defined as $B\defeq \{\phis \in \Phisn~|~\dfreq{\phis} \subseteq \fsub \}$, then the cardinality of set $B$ is upper bounded by $(n+1)^{|\fsub|}$.
	\end{lemma}
\begin{proof}[Proof of \Cref{thmbeta}]
	Using the law of total probability we have,
	\begin{align*}
	\begin{split}
	\Prob{|\bff_{{S}}(\bp)-\bff_{{S}}(\bpml)|\geq 2\epsilon} &=\Prob{|\bff_{{S}}(\bp)-\bff_{{S}}(\bpml)|\geq 2\epsilon~,~\dfreq{\phis} \subseteq \fsub}\\
	&\quad + \Prob{ |\bff_{{S}}(\bp)-\bff_{{S}}(\bpml)|\geq 2\epsilon~,~\dfreq{\phis} \not\subseteq \fsub},\\
	& \leq \Prob{|\bff_{{S}}(\bp)-\bff_{{S}}(\bpml)|\geq 2\epsilon~,~\dfreq{\phis} \subseteq \fsub}  + \Prob{\dfreq{\phis} \not\subseteq \fsub}~.
	\end{split}
	\end{align*}
	
	Consider any $\phis \sim \bp$. If $\probbp{\phis}>\delta/\beta$, then we know that $\probbpml{\phis}>\delta$. For $\beta \leq 1$, we have $\probbp{\phis}>\delta$ that implies $|\bff_{S}(\bp)-\estiS{\phis}| \leq \epsilon$. Further $\probbpml{\phis}>\delta$ implies $|\bff_{S}(\bpml)-\estiS{\phis}| \leq \epsilon$. Using triangle inequality we get,
	$|\bff_{S}(\bp)-\bff_{S}(\bpml)| \leq |\bff_{S}(\bp)-\estiS{\phis}|+|\bff_{S}(\bpml)-\estiS{\phis}| \leq 2\epsilon$.
	Note we wish to upper bound the probability of set: $\bset\defeq \{\phis \in \Phisn~|~\dfreq{\phis} \subseteq \fsub \text{ and }|\bff_{{S}}(\bp)-\bff_{{S}}(\bpml)|\geq 2 \epsilon\}$. From the previous discussion, we get $\probbp{\phis} \leq \delta/\beta$ for all $\phis \in \bset$. Therefore,
	\begin{align*}
	\Prob{|\bff_{{S}}(\bp)-\bff_{{S}}(\bpml)|\geq 2\epsilon~,~\dfreq{\phis} \subseteq \fsub}=\sum_{\phis \in \bset }\probbp{\phis} \leq \frac{\delta}{\beta}|\bset| \leq \frac{\delta}{\beta} (n+1)^{|\fsub|}~.
	\end{align*}
	In the final inequality, we use $\bset \subseteq \{\phis \in \Phisn~|~\dfreq{\phis} \subseteq \fsub\}$ and invoke \Cref{lem:bcard}.
\end{proof}

\begin{proof}[Proof for \Cref{thm:main}]
Using Bayes rule we have:
	\begin{equation}
	\begin{split}
	\Prob{|\bff_{S}(\bp)-\bff_{S}(\bpml)|>2\epsilon}&=\sum_{S'\subseteq \bX} \Prob{|\bff_{S}(\bp)-\bff_{S}(\bpml)|>2\epsilon~|~S=S'}\Prob{S=S'}\\
	&\leq \sum_{S' \in 2^\gset} \Prob{|\bff_{S}(\bp)-\bff_{S}(\bpml)|>2\epsilon~|~S=S'}\Prob{S=S'}+ \Prob{S \notin 2^\gset}~.
	\end{split}
	\end{equation}
	In the second inequality, we use $\sum_{S' \notin 2^\gset} \Prob{|\bff_{S}(\bp)-\bff_{S}(\bpml)|>2\epsilon ~,~ S=S'} \leq \Prob{S \notin 2^\gset}$. Consider the first term on the right side of the above expression and note that it is upper bounded by, $\sum_{S' \in 2^\gset} \Prob{|\bff_{S'}(\bp)-\bff_{S'}(\bpmlp)|>2\epsilon}\Prob{S=S'}
\leq \sum_{S' \in 2^\gset} \left[\frac{\delta n^{|F'|}}{\beta}+\Prob{\dfreq{\phisp} \not\subseteq \fsubp} \right] \Prob{S=S'}
\leq \frac{\delta n^{|F'|}}{\beta}+\gamma$. % + \Prob{S \notin 2^\gset}
	In the first upper bound, we removed randomness associated with the random set $S$ and used $\Prob{|\bff_{S}(\bp)-\bff_{S}(\bpml)|>2\epsilon~|~S=S'}= \Prob{|\bff_{S'}(\bp)-\bff_{S'}(\bpmlp)|>2\epsilon}$.
	In the first inequality above, we invoke \Cref{thmbeta} using conditions from \Cref{eq:condio}. In the second inequality, we use $\sum_{S' \in 2^\gset} \Prob{S=S'} \leq 1$ and $\Prob{\dfreq{\phis} \not\subseteq \fsubp, S=S'} \leq \gamma$. The theorem follows by combining all the analysis together.
\end{proof}

 \section{Applications of the General Framework}\label{sec:appl}
Here we provide applications of our general framework (defined in \Cref{sec:results}) using results from the previous section. We apply our general framework to estimate entropy and distance to uniformity. In \Cref{sec:entr} and \Cref{sec:dtu} we analyze the performance of our estimator for entropy and distance to uniformity estimation respectively.
\subsection{Entropy estimation}\label{sec:entr}
In order to prove our main result for entropy (\Cref{thm:entr}), we first need the existence of an estimator for entropy with some desired properties. The existence of such an estimator will be crucial to bound the failure probability of our estimator. A result analogous to this is already known in \cite{ADOS16} (Lemma 2) and the proof of our result follows from a careful observation of \cite{ADOS16, WY16}. We state this result here but defer the proof to appendix.
\begin{lemma}\label{thm:entrochange}
	Let $\alpha>0$, $\epsilon>\Omega\left(\frac{\log N}{N^{1-\alpha}}\right)$ and $S \subseteq \bX$, then for entropy on subset $S$ ($\sum_{y\in S}\bp_{y}\log \frac{1}{\bp_{y}}$) there exists an $S$-pseudo profile based estimator that use the optimal number of samples, has bias less than $\epsilon$ and if we change any sample, changes by at most $c\cdot \frac{n^{\alpha}}{n}$, where $c$ is a constant.
\end{lemma}
Combining the above lemma with \Cref{thm:main}, we next prove that our estimator defined in \Cref{algpml} is sample complexity optimal for estimating entropy in a broader regime of error $\epsilon$.
\begin{proof}[Proof for \Cref{thm:entr}]
	Let $\bff(\bp)$ represent the entropy of distribution $\bp$ and $\hat{\bff}$ be the estimator in \Cref{thm:entrochange}.
	Define $\fsub \defeq [0, \co \log n]$ for constant $c_{1}\geq 40$. Given the sequence $x^{2n}$, the random set $S$ is defined as $S\defeq \{y \in \bX~|~f(\xon,y) \leq \co \log n \}$.
	Let $\fsubp\defeq[0,8\co \log n]$, then by derivation in Lemma 6~\cite{ADOS16} (or by simple application of Chernoff \footnote{Note probability of many events in this proof can be easily bounded by application of Chernoff. These bounds on probabilities are also shown in \cite{ADOS16, WY16} and we use these inequalities by omitting details.}) we have,
	$$ \Prob{\dfreq{\phis} \not\subseteq \fsubp} =\Prob{\exists y\in \bX \text{ such that }\bff(\xon,y)\leq \co \log n \text{ and }\bff(\xtn,y)> 8\co \log n}\leq \frac{1}{n^{5}}~.$$
	Further let $\gset\defeq \{x \in \bX ~|~\bp_{x}\leq \frac{\consto \log N}{n}\}$, then by Equation 48 in \cite{WY16} we have, $\Prob{S \notin 2^\gset} \leq \frac{1}{n^4}$.
	Further for all $S' \in 2^{\gset}$ we have,
	$$\Prob{\dfreq{\phi_{S'}} \not\subseteq \fsubp}=\Prob{\exists y\in S' \text{ such that }\bff(\xtn,y)> 8\co \log n}\leq \gamma \text{ for } \gamma=\frac{1}{n^{5}}~.$$
	Note for all $x \in S'$, $\bp_{x} \leq \frac{\consto \log N}{n}$ and the above inequality also follows from Chernoff. All that remains now is to upper bound $\delta$. Using the estimator constructed in \Cref{thm:entrochange} and further combined with McDiarmid’s inequality, we have,
	$$\Prob{|\bff_{{S'}}(\bp)-\hat{\bff}(\phisp)|\geq 2 \epsilon} \leq 2\expo{\frac{-2\epsilon^2}{n(c\frac{n^\alpha}{n})^2}} \leq \delta \text{ for }\delta= \expo{-2\epsilon^2n^{1-2\alpha}}~.$$
Substituting all these parameters together in \Cref{thm:main} we have,
	\begin{equation}
	\begin{split}
	\Prob{|\bff_{S}(\bp)-\bff_{S}(\bpml)|>2\epsilon}&\leq \frac{\delta n^{|F'|}}{\beta}+\Prob{\dfreq{\phis} \not\subseteq \fsubp} + \Prob{S \notin 2^\gset}\\
	&\leq  \expo{-2\epsilon^2n^{1-2\alpha}} n^{9 \co \log n}+\frac{1}{n^4} \leq \frac{2}{n^4}~.
	\end{split}
	\end{equation}
	In the first inequality, we use \Cref{thm:main}. In the second inequality, we substituted the values for $\delta, \gamma, \beta$ and $\Prob{S \notin 2^\gset}$. In the final inequality we used $n=\Theta(\frac{N}{\log N}\frac{1}{\epsilon})$ and $\epsilon>\Omega\left(\frac{\log^3 N}{N^{1-4\alpha}}\right)$.
	
		Our final goal is to estimate $\bff(\bp)$, and to complete the proof we need to argue that $\bff_{\bar{S}}(\hat{\bp})$ + the correction bias with respect to $\bff_{\bar{S}}$ is close to $\bff_{\bar{S}}(\bp)$, where recall $\hat{\bp}$ is the empirical distribution on sequence $\xtn$. The proof for this follows immediately from \cite{WY16} (Case 2 in the proof of Proposition 4). \cite{WY16} bound the bias and variance of the empirical estimator with a correction bias and applying Markov inequality on their result we get $\Prob{|\bff_{\bar{S}}(\bp)-(\bff_{\bar{S}}(\hat{\bp})+\frac{|\bar{S}|}{n})|>2\epsilon} \leq \frac{1}{3}$, where $\frac{|\bar{S}|}{n}$ is the correction bias in \cite{WY16}. Using triangle inequality, our estimator fails if either $|\bff_{\bar{S}}(\bp)-(\bff_{\bar{S}}(\hat{\bp})+\frac{|\bar{S}|}{n})|>2\epsilon$ or $|\bff_{S}(\bp)-\bff_{S}(\bpml)|>2\epsilon$. Further by union bound the failure probability is at most $\frac{1}{3}+\frac{2}{n^4}$, which is a constant.
\end{proof}
\subsection{Distance to Uniformity estimation}\label{sec:dtu}
Here we prove our main result for distance to uniformity estimation (\Cref{thm:dtu}). First, we show existence of an estimator for distance to uniformity with certain desired properties. Similar to entropy, a result analogous to this is shown in \cite{ADOS16} (Lemma 2) and the proof of our result follows from the careful observation of \cite{ADOS16, JHW17}. We state this result here but defer the proof to \Cref{sec:omitted}.
\begin{lemma}\label{thm:dtuchange}
	Let $\alpha>0$ and $S \subseteq \bX$, then for distance to uniformity on $S$ ($\sum_{y \in S}|\bp_{y}-\frac{1}{N}|$) there exists an $S$-pseudo profile based estimator that use the optimal number of samples, has bias at most $ \epsilon$ and if we change any sample, changes by at most $c\cdot \frac{n^{\alpha}}{n}$, where $c$ is a constant.
\end{lemma}
Combining the above lemma with \Cref{thm:main} we provide the proof for \Cref{thm:dtu}.
\begin{proof}[Proof for \Cref{thm:dtu}]
	Let $\bff(\bp)$ represent the distance to uniformity for distribution $\bp$ and $\hat{\bff}$ be the estimator in \Cref{thm:dtuchange}.
	Define $\fsub = [\frac{n}{N}-\sqrt{\frac{c_{1}n \log n}{N}}, \frac{n}{N}+\sqrt{\frac{c_{1}n \log n}{N}}]$ for some constant $c_1 \geq 40$. Given the sequence $x^{2n}$, the random set $S$ is defined as $S\defeq \{y \in \bX~|~f(\xon,y) \in \fsub \}$. Let $\fsubp=[\frac{n}{N}-\sqrt{\frac{8c_{1}n \log n}{N}}, \frac{n}{N}+\sqrt{\frac{8c_{1} n\log n}{N}}]$, then by derivation in Lemma 7 of \cite{ADOS16} (also shown in \cite{JHW17} \footnote{Similar to entropy, for many events their probabilities can be bounded by simple application of Chernoff and have already been shown in \cite{ADOS16, JHW17}. We omit details for these inequalities.}) we have,
	$$ \Prob{\dfreq{\phis} \not\subseteq \fsubp} =\Prob{\exists y\in \bX \text{ such that }\bff(\xon,y)\in \fsub \text{ and }\bff(\xtn,y) \notin \fsubp}\leq \frac{1}{n^{4}}~.$$
	 Further let $\gset\defeq \{x \in \bX ~|~\bp_{x}\in [\frac{1}{N}-\sqrt{\frac{2c_{1} \log n}{nN}}, \frac{1}{N}+\sqrt{\frac{2c_{1} \log n}{nN}}]\}$, then using Lemma 2 in \cite{JHW17} we get,
	$$\Prob{S \notin 2^\gset}=\Prob{\exists y \in \bX \text{ such that }\bff(\xon,y)\in \fsub \text{ and }\bp_{x} \notin G} \leq \frac{\log n}{n^{1-\epsilon}}~.$$
	Further for all $S' \in 2^{\gset}$ we have,
	$$\Prob{\dfreq{\phi_{S'}} \not\subseteq \fsubp}=\Prob{\exists y\in S' \text{ such that }\bff(\xtn,y)> 8\co \log n}\leq \gamma \text{ for } \gamma=\frac{1}{n}~.$$
	Note for all $x \in S'$, $\bp_{x} \in G$ and the above result follows from \cite{JHW17} (Lemma 1). All that remains now is to upper bound $\delta$. Using the estimator constructed in \Cref{thm:dtuchange} and further combined with McDiarmid’s inequality, we have,
	$$\Prob{|\bff_{{S'}}(\bp)-\hat{\bff}(\phisp)|\geq 2 \epsilon} \leq 2\expo{\frac{-2\epsilon^2}{n(c\frac{n^\alpha}{n})^2}} \leq \delta \text{ for }\delta= \expo{-2\epsilon^2n^{1-2\alpha}}~.$$
	Substituting all these parameters in \Cref{thm:main} we get,
	\begin{equation}\label{eq:entropml}
	\begin{split}
	\Prob{|\bff_{S}(\bp)-\bff_{S}(\bpml)|>2\epsilon}&\leq \frac{\delta n^{|F'|}}{\beta}+\Prob{\dfreq{\phis} \not\subseteq \fsubp} + \Prob{S \notin 2^\gset}\\
	&\leq  \expo{-2\epsilon^2n^{1-2\alpha}} n^{2\sqrt{\frac{8c_{1}n \log n}{N}}}+\frac{\log n}{n^{1-\epsilon}}+\frac{1}{n} \leq o(1)~.
	\end{split}
	\end{equation}
	In the first inequality, we use \Cref{thm:main}. In the second inequality, we substituted values for $\delta, \gamma, \beta$ and $\Prob{S \notin 2^\gset}$. In the final inequality we used $n=\Theta(\frac{N}{\log N}\frac{1}{\epsilon^2})$ and $\epsilon>\Omega\left(\frac{1}{N^{1-8\alpha}}\right)$.
	
	Our final goal is to estimate $\bff(\bp)$, and to complete the proof we argue that $\bff_{\bar{S}}(\hat{\bp})$ + correction bias with respect to $\bff_{\bar{S}}$ is close to $\bff_{\bar{S}}(\bp)$, where recall $\hat{\bp}$ is the empirical distribution on sequence $\xtn$. The proof for this case follows immediately from \cite{JHW17} (proof of Theorem 2). \cite{JHW17} define three kinds of events $\mathcal{E}_{1},\mathcal{E}_{2}$ and $\mathcal{E}_{3}$, the proof for our empirical case follows from the analysis of bias and variance of events $\mathcal{E}_{1}$ and $\mathcal{E}_{2}$. Further combining results in \cite{JHW17} with Markov inequality we get $\Prob{|\bff_{\bar{S}}(\bp)-\bff_{\bar{S}}(\hat{\bp})|>2\epsilon} \leq \frac{1}{3}$, and the correction bias here is zero. Using triangle inequality, our estimator fails if either $|\bff_{\bar{S}}(\bp)-(\bff_{\bar{S}}(\hat{\bp})+\frac{|\bar{S}|}{n})|>2\epsilon$ or $|\bff_{S}(\bp)-\bff_{S}(\bpml)|>2\epsilon$. Further by union bound the failure probability is upper bounded by $\frac{1}{3}+o(1)$, which is a constant.
\end{proof}
\section{Experiments}\label{sec:exp}
We performed two different sets of experiments for entropy estimation -- one to compare performance guarantees and the other to compare running times. In our pseudo PML approach, we divide the samples into two parts. We run the empirical estimate on one (this is easy) and the PML estimate on the other. For the PML estimate, any algorithm to compute an approximate PML distribution can be used in a black box fashion. 
An advantage of the pseudo PML approach is that it can use any algorithm to estimate the PML distribution as a black box, providing both competitive performance and running time efficiency. 
In our experiments, we use the heuristic algorithm in \cite{PJW17} to compute an approximate PML distribution. In the first set of experiments detailed below, we compare the performance of the pseudo PML approach with raw \cite{PJW17} and other state-of-the-art estimators for estimating entropy. Our code is available at \url{https://github.com/shiragur/CodeForPseudoPML.git}
\begin{figure}[!ht]
	\centering
	\includegraphics[width=\linewidth]{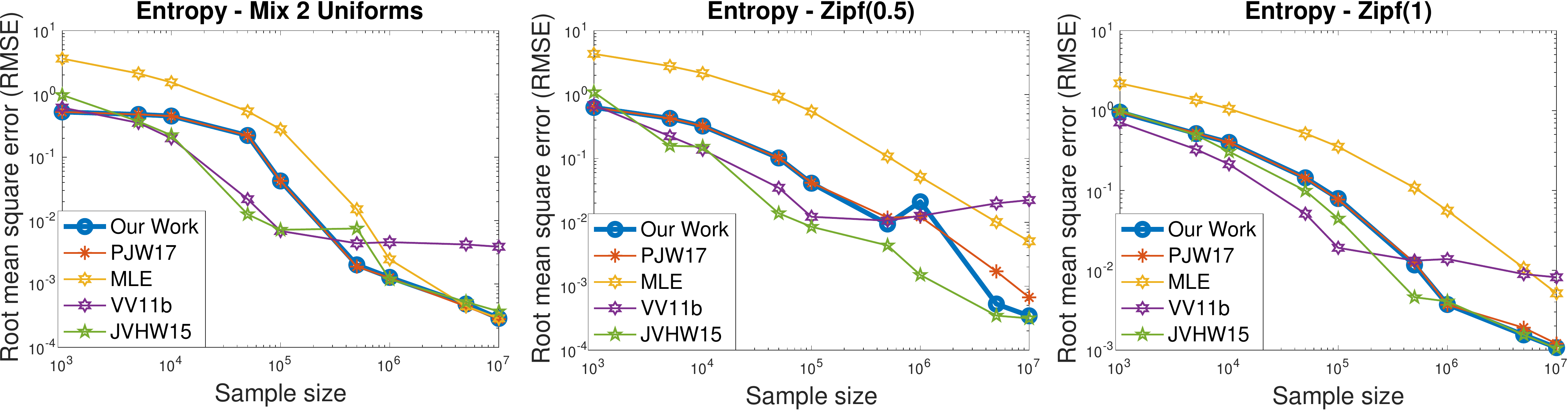}
	\label{fig:entr}\end{figure}
\vspace{-24pt}

Each plot depicts the performance of various algorithms for estimating entropy of different distributions with domain size $N=10^5$. 
Each data point represents 50 random trials. ``Mix 2 Uniforms'' is a mixture of two uniform distributions, with half the probability mass on
the first $N/10$ symbols, and $\mathrm{Zipf}(\alpha) \sim 1/i^{\alpha}$ with $i \in [N]$. 
MLE is the naive approach of using the empirical distribution with correction bias;
all the remaining algorithms are denoted using bibliographic citations.
In our algorithm we pick $threshold=18$ (same as \cite{WY16}) and our set $\mathrm{F}=[0,18]$ (input of \Cref{algpml}), i.e. we use the PML estimate on frequencies $\leq 18$ and empirical estimate on the rest. Unlike \Cref{algpml}, we do not perform sample splitting in the experiments -- we believe this requirement is an artifact of our analysis. For estimating entropy, the error achieved by our estimator is competitive with \cite{PJW17} and other state-of-the-art entropy estimators. Note that our results match \cite{PJW17} for small sample sizes because not many domain elements cross the threshold and for a large fraction of the samples, we simply run the \cite{PJW17} algorithm.

In the second set of experiments we demonstrate the running time efficiency of our approach. In these experiments, we compare the running time of our algorithm using \cite{PJW17} as a subroutine to the raw \cite{PJW17} algorithm on the $\mathrm{Zipf(1)}$ distribution. The second row is the fraction of samples on which our algorithm uses the empirical estimate (plus correction bias). The third row is the ratio of the running time of \cite{PJW17} to our algorithm. 
For large sample sizes, the entries in the EmpFrac row have high value, i.e. our algorithm applies the simple empirical estimate on large fraction of samples; therefore, enabling $10$x speedup in the running times. 
\vspace*{-0.1in}
\begin{center}
	\begin{tabular}{|c||c|c|c|c|c|c|c|c|} 
		\hline
		Samples size & $10^3$ & $5*10^3$ & $10^4$ & $5*10^4$ & $10^5$ & $5*10^5$ &$10^6$ & $5*10^6$ \\ [0.5ex]   
		\hline
		EmpFrac & 0.184 & 0.317 & 0.372 & 0.505 & 0.562 & 0.695 & 0.752 & 0.886 \\  
		\hline
		Speedup & 0.824 & 1.205 & 1.669 & 3.561 & 4.852 & 9.552 & 13.337 & 12.196 \\
		\hline
	\end{tabular}
\end{center}
\subsubsection*{Acknowledgments}
We thank the reviewers for the helpful comments, great suggestions, and positive feedback.
Moses Charikar was supported by a Simons Investigator Award, a Google Faculty Research Award and an Amazon Research Award.
Aaron Sidford was partially supported by NSF CAREER Award CCF-1844855.
\bibliographystyle{alpha}
\bibliography{PML}
\appendix
\newcommand{\pml}{\bp_{\phi}}
\newcommand{\delem}{\mathrm{Distinct}}
\section{Support Estimation}\label{sec:support}
Here we study the PML based plug-in estimator for support estimation. \cite{ADOS16} showed that PML based plug-in estimator is sample complexity optimal for estimating support within additive accuracy $\epsilon \lprob$ for all $\epsilon>\frac{1}{\lprob^{0.2499}}$. Further for any $\epsilon<\frac{1}{\lprob^{\delta}}$ for some constant $\delta>0$, the empirical distribution based plug-in estimator is exact with high probability. Here we provide proofs for two main results described in \Cref{sec:results} for support. In \Cref{thm:supp} and \Cref{thm:approxsupp}, we show that PML and approximate PML distributions (under the constraint that all its probability values are $\geq \frac{1}{k}$) based plug-in estimators are sample complexity optimal for all parameter regimes, thus providing a better analysis for \cite{ADOS16}.

We next define a function that outputs the number of distinct frequencies in the profile. Later in \Cref{lem:pmlprob}, we show that the support of PML and approximate PML distribution is at least the number of distinct elements in the sequence.
\begin{defn}
	For any $S \subseteq \bX$, the function $\delem:\Phi^{n}\rightarrow \Z$, takes input $\phi$ and returns $\sum_{j\in[n]}\phi_j$. For any sequence $x^n$, we overload notation and use $\delem(x^n)$ to denote $\delem(\Phi(x^n))$. Note $\delem(\phi)$ and $\delem(x^n)$ denote the number of distinct domain elements observed in profile $\phi$ or sequence $x^{n}$ respectively.
	\end{defn}
\begin{lemma}\label{lem:pmlprob}
	For any distribution $\bp \in \simplex$ such that $\bp_x \in \{0\} \cup [\frac{1}{k},1]$ and a profile $\phi \in \Phi^{n}$, if $S(\bp)<\delem(\phi)$ then $\bp(\phi)=0$.
\end{lemma}
\begin{proof}
Consider sequences $x^{n}$ with $\Phi(x^n)=\phi$. All such sequences have $\delem(\phi)$ number of distinct observed elements that is strictly greater than $S(\bp)$ and distribution $\bp$ assigns probability zero for all these sequences.
\end{proof}
\newcommand{\suppof}{\mathrm{Support}}
\begin{proof}[Proof for \Cref{thm:supp}]
	Given $\phi$, let $\pml \in \simplex$, be the distribution with $\pml(x) \in \{0\}\cup[\frac{1}{k},1]$. If $\epsilon>\frac{1}{\lprob^{0.2499}}$, we know that plug-in approach on $\pml$ is sample complexity optimal~\cite{ADOS16}. We consider the regime where $\epsilon \leq \frac{1}{\lprob^{0.2499}}$ and here the number of samples $n=c \cdot \lprob\log \lprob$ for some constant $c \geq 2$. If $S(\pml) < \delem(\phi)$, then by \Cref{lem:pmlprob} we have $\pml(\phi)=0$ a contradiction because the empirical distribution assigns a non-zero probability value for observing $\phi$. Therefore, without loss of generality we assume $S(\pml) \geq \delem(\phi)$. We next argue that $S(\pml) = \delem(\phi)$. We prove this statement by contradiction. Suppose $S(\pml) > \delem(\phi)$, then define $\pml' \in \simplex$ to be the PML distribution under constraints $S(\pml')=\delem(\phi)$ and $\pml'(x) \in \{0\}\cup[\frac{1}{k},1]$. Let $\suppof:\simplex \rightarrow 2^{\bX}$ be a function that takes distribution $\bp$ as input and returns index set for the support of $\bp$. Now consider $\Prob{\pml,\phi}$, and recall $\Prob{\pml,\phi}=\sum_{\{x^n\in \bX^n~|~\Phi(x^n)=\phi\}}\Prob{\pml,x^n}$. Further note that $\sum_{\{x^n\in \bX^n~|~\Phi(x^n)=\phi\}}\Prob{\pml,x^n}=\sum_{\{S\subseteq \suppof(\pml)~|~|S|=\delem(\phi) \}} \sum_{\{x^n\in S^n~|~\Phi(x^n)=\phi\}}\Prob{\pml,x^n}$, therefore,
\begin{equation}\label{eq:a1}
\begin{split}
\Prob{\pml,\phi}&=\sum_{\{S\subseteq \suppof(\pml)~|~|S|=\delem(\phi) \}} \sum_{\{x^n\in S^n~|~\Phi(x^n)=\phi\}}\Prob{\pml,x^n}\\
& \leq \sum_{\{S\subseteq \suppof(\pml)~|~|S|=\delem(\phi) \}} \left(1-\frac{S(\pml)-|S|}{k} \right)^n \Prob{\pml',\phi}~.
\end{split}
\end{equation}
\newcommand{\pmls}{\bp_{\phi,S}}
In the second inequality, we use for all $x \in \bX$, $\pml(x)\in \{0\}\cup [\frac{1}{k},1]$ and we have $\sum_{x\in S}\pml(x) \leq \left(1-\frac{S(\pml)-|S|}{k} \right)$ and the inequality follows. 

We next upper bound the term $\sum_{\{S\subseteq \suppof(\pml)~|~|S|=\delem(\phi) \}} \left(1-\frac{S(\pml)-|S|}{k} \right)^n$. Note that, 
\begin{align*}
\sum_{\{S\subseteq \suppof(\pml)~|~|S|=\delem(\phi) \}} & \left(1-\frac{S(\pml)-|S|}{k} \right)^n  \leq \expo{-n\frac{S(\pml)-\delem(\phi)}{k}} \binom{S(\pml)}{\delem(\phi)} \\
& \leq  \expo{-n\frac{S(\pml)-\delem(\phi)}{k} + (S(\pml)-\delem(\phi))\log S(\pml) }\\
& \leq   \expo{ -\log k }~.
\end{align*}
In the second inequality, we use a weak upper bound on the quantity $\binom{S(\pml)}{\delem(\phi)}$. In the third and fourth inequality, we use $n=ck\log k$, $c\geq2$ and $k \geq  S(\pml) > \delem(\phi)$. Combining everything together we get, $\Prob{\pml,\phi} \leq \expo{ -\log k }\Prob{\pml',\phi}$. A contradiction because $\pml$ is the PML distribution.

Therefore if $n>2k\log k$, then the previous derivation implies, 
$$\Prob{S(\pml)=\delem(\phi)}=1~.$$ 
Further if $n>2k\log k$, then 
$$\Prob{S(\bp) =\delem(\phi)} \geq 1-k\expo{\frac{-n}{k}}~.$$ 
Combining previous two inequalities and substituting $n>2k\log k$ we get, $\Prob{S(\bp) =S(\pml)} \geq 1-\exps{-\log k}$, thus concluding the proof.
\end{proof}

\renewcommand{\bpml}{\bp^{\beta}_{\phi}}
\begin{proof}[Proof for \Cref{thm:approxsupp}]
		The proof for this result is similar to \Cref{thm:supp} and for completeness we reprove it. Given $\phi$, let $\pml,\bpml \in \simplex$, be PML and $\beta$-approximate PML distributions respectively under the constraint $\pml(x),\bpml(x) \in \{0\}\cup[\frac{1}{k},1]$. If $\epsilon>\frac{1}{\lprob^{0.2499}}$, by \cite{ADOS16} we already know that plug-in approach on $\bpml$ for $\beta=\exps{-\epsilon^2 n^{1-\alpha}}$ is sample complexity optimal with high probability. Here we consider the regime $\epsilon \leq \frac{1}{\lprob^{0.2499}}$ and in this case the number of samples $n=c \cdot \lprob\log \lprob$ for some large constant $c \geq 2$. If $S(\bpml) < \delem(\phi)$, then by \Cref{lem:pmlprob} we have $\bpml(\phi)=0$ which is a contradiction, because the empirical distribution clearly returns a non-zero probability value. Therefore, without loss of generality we assume $S(\bpml) \geq \delem(\phi)$. We next argue that $S(\bpml) \leq \delem(\phi)+\epsilon k$. We prove this statement by contradiction. Suppose $S(\bpml) > \delem(\phi)+\epsilon k$, then consider the $\Prob{\bpml,\phi}$, and recall $\Prob{\bpml,\phi}=\sum_{\{x^n\in \bX^n|\Phi(x^n)=\phi\}}\Prob{\bpml,x^n}$. Further note that $\sum_{\{x^n\in \bX^n|\Phi(x^n)=\phi\}}\Prob{\bpml,x^n}=\sum_{\{S\subseteq \suppof(\bpml)||S|=\delem(\phi) \}} \sum_{\{x^n\in S^n|\Phi(x^n)=\phi\}}\Prob{\bpml,x^n}$. Therefore,
	\begin{equation}
	\begin{split}
	\Prob{\bpml,\phi}&=\sum_{\{S\subseteq \suppof(\bpml)||S|=\delem(\phi) \}} \sum_{\{x^n\in S^n|\Phi(x^n)=\phi\}}\Prob{\bpml,x^n}\\
	& \leq \sum_{\{S\subseteq \suppof(\bpml)~|~|S|=\delem(\phi) \}} \left(1-\frac{S(\bpml)-|S|}{k} \right)^n \Prob{\pml,\phi}~.
	\end{split}
	\end{equation}
	In the final inequality we used for all $x \in \bX$, $\bpml(x)\in \{0\}\cup [\frac{1}{k},1]$ and we have $\sum_{x\in S}\bpml(x) \leq \left(1-\frac{S(\bpml)-|S|}{k} \right)$ and using the definition of $\pml$ the inequality follows. 
	
	We next upper bound the term $\sum_{\{S\subseteq \suppof(\bpml)~|~|S|=\delem(\phi) \}} \left(1-\frac{S(\bpml)-|S|}{k} \right)^n$. Note that, 
	\begin{align*}
	\sum_{\{S\subseteq \suppof(\bpml)~|~|S|=\delem(\phi) \}} & (1-\frac{S(\bpml)-|S|}{k} )^n \leq \exps{-n\frac{S(\bpml)-\delem(\phi)}{k}} \binom{S(\bpml)}{\delem(\phi)} \\
	& \leq \exps{-n\frac{S(\bpml)-\delem(\phi)}{k} + (S(\bpml)-\delem(\phi))\log S(\bpml) } \\
	& \leq  \exps{ (S(\bpml)-\delem(\phi))(\log k- c \log k } \\
	& \leq  \exps{ -\epsilon k \log k } < \exps{ -\epsilon^2 n^{1-4\alpha} }~.
	\end{align*}
	\newcommand{\pmls}{\bp_{\phi,S}}
	 In the second inequality, we use a weak upper bound for the quantity $\binom{S(\pml)}{\delem(\phi)}$. In the third and fourth inequality, we use $n=ck\log k$, $c\geq2$ and $S(\bpml) > \delem(\phi)+\epsilon k$. In the final inequality, we use $n^{1-\alpha}\leq k\log k$ for constant $\alpha>0$. Combining everything together we get $\Prob{\bpml,\phi} < \exps{ -\epsilon^2 n^{1-4\alpha} } \Prob{\pml,\phi}$, a contradiction on the definition of $\bpml$.
	
	Therefore if $n>2k\log k$, then the previous derivation implies $\Prob{|S(\bpml)-\delem(\phi)|\geq \epsilon k}=1$. Further if $n>2k\log k$, then $\Prob{S(\bp) =\delem(\phi)} \geq 1-k\exps{\frac{-n}{k}}$. Combining the previous two inequalities and substituting $n>2k\log k$ we get, $\Prob{|S(\bpml)-S(\bp)|\geq \epsilon k} \geq 1-\exps{-\log k}$, thus concluding the proof.
	\end{proof}
\newcommand{\dP}{\bp}
\newcommand{\smb}{y}
\newcommand{\dPsmb}{\bp_{\smb}}
\newcommand{\absv}[1]{| #1 |}
\newcommand{\Mltsmb}{n_{\smb}}
\newcommand{\absz}{N}
\newcommand{\nsmp}{n}
\newcommand{\NN}{\N}
\newcommand{\cO}{O}
\newcommand{\dst}{\epsilon}
\newcommand{\probof}[1]{\Prob{#1}}
\newcommand{\flnpwrss}[2]{#1^{\underline{#2}}}

\section{Omitted Proof from \Cref{app:universal}}
Here we provide the proof for \Cref{lem:bcard}.
\begin{proof}[Proof for \Cref{lem:bcard}] Fix an ordering on the elements of $\fsub$. Let $\fsub(i)$ denote the $i$'th frequency element of $\fsub$. For all $\phis \in B$, the set of distinct frequencies in $\phis$ is a subset of $\fsub$ and the length of $\phis$ is equal to $n$. Therefore, any element $\phis \in B$ can be encoded as a unique vector $v_{\phis} \in [0,n]^{\fsub}$, where $v_{\phis}(i) \defeq \F(\fsub(i))$ denotes the number of elements in $\phis$ that have frequency $F(i)$. Using the previous discussion, we have $|B| \leq |[0,n]^{\fsub}| \leq (n+1)^{|\fsub|}$.
\end{proof}

\section{Omitted Proofs from \Cref{sec:appl}}\label{sec:omitted}
Here, we present and prove results related to the existence of an estimator for entropy and distance to uniformity on a fixed subset $S\subseteq \bX$. Note the estimator we provide here is exactly same to the one presented in~\cite{ADOS16} but defined only on subset $S\subseteq \bX$. All the results and proofs presented here are similar to the ones in~\cite{ADOS16} and for completeness and verification purposes we reprove (with slight modifications) these results. As in~\cite{ADOS16}, we first provide a general definition of an estimator that works both for entropy and distance to uniformity. In \Cref{lem:general_est}, we prove a result that captures the maximum change of this general estimator by changing one sample. In section \ref{subsec:entr} and \ref{subsec:dtu}, we provide proofs for entropy and distance to uniformity respectively.

Given $2n$ samples $x^{2n}=(\xon,\xtn)$ from distribution $\bp$. Let
$\Mltsmb^{'}\defeq \bff(\xon,y)$, and $\Mltsmb\defeq \bff(\xtn,y)$ be the number of appearances of symbol $\smb$ in the first and second half respectively. We define the following estimator which is exactly the same as \cite{ADOS16} but defined only on subset $S$. For all $x \in S$,
\[
\hat{g}_{S}(x^{2n})
=\max\left\{\min\left\{\sum_{\smb \in S}{g_{\smb}},f_{S,\max}\right\}, 0\right\}.
\]
where $f_{S,\max}$ is the maximum value of the property $f$ on subset $S$ and for all $\smb \in S$,
\[
g_{\smb}=
\begin{cases}
G_{L,g}(\Mltsmb), & \text{ for } \Mltsmb^{'}<c_2\log N, \text{ and }
\Mltsmb< c_1\log N,\\ 0, & \text{ for } \Mltsmb^{'}<c_2\log k, \text{
  and } \Mltsmb\ge c_1\log N,\\ g \left(\frac{\Mltsmb}{n}\right)+ g_n
, & \text{ for } \Mltsmb^{'}\ge c_2\log N,
\end{cases}
\]
where $g_n$ is the first order bias correction term for $g$,
$G_{L,g}(\Mltsmb) = \sum^L_{i=1} b_i {(\frac{\Mltsmb}{n})}^{i}$ is
the unbiased estimator for $P_{L,g}(\bp_{y})$, the optimal uniform approximation of function $g$ by degree-$L$ polynomials on $[0; \co \frac{\log n}{n}]$. 
\begin{lemma}
  \label{lem:general_est}
  For any estimator $\hat{g}$ defined as above, changing any one of
  the sample changes the estimator by at most
  \[
  9\max\left(e^{L^2/n}\max
  |b_i|, \frac{L_g}{n}, g\left(\frac{c_1 \log (n)}n\right), g_n\right),
  \]
  where $L_g = n \max_{i\in \NN} |g(i/n)- g((i-1)/n)|$.
\end{lemma}
\begin{proof}
Given $2n$ samples $x^{2n}=(\xon,\xtn)$ from distribution $\bp$. Recall the estimator for entropy and distance to uniform from \cite{ADOS16},
\[
\hat{g}(x^{2n})
=\max\left\{\min\left\{\sum_{\smb \in S}{g_{\smb}},f_{\max}\right\}, 0\right\}.
\]
where $f_{\max}$ is the maximum value of the property $f$ and for all $\smb \in \bX$,
\[
g_{\smb}=
\begin{cases}
G_{L,g}(\Mltsmb), & \text{ for } \Mltsmb^{'}<c_2\log N, \text{ and }
\Mltsmb< c_1\log N,\\ 0, & \text{ for } \Mltsmb^{'}<c_2\log k, \text{
	and } \Mltsmb\ge c_1\log N,\\ g \left(\frac{\Mltsmb}{n}\right)+ g_n
, & \text{ for } \Mltsmb^{'}\ge c_2\log N,
\end{cases}
\]
Now construct a new sequence from $x^n$ as follows: replace all symbols in $\bar{S}$ (appearing in $x^n$) by a unique symbol $y' \in \bar{S}$ and call this new sequence $z^n$. Now note our estimator is unaffected by this change, because it only depends on the occurrences of elements in $S$. The change in the value of estimator in~\cite{ADOS16} by changing one sample in $z^n$ is upper bounded by:
\begin{equation}
8\max\left(e^{L^2/n}\max
|b_i|, \frac{L_g}{n}, g\left(\frac{c_1 \log (n)}n\right), g_n\right),
\end{equation}
The above result follows by Lemma 5 in~\cite{ADOS16}. We next study the change in the value of our estimator by changing one sample in $z^n$. Note this is equivalent to the change in the value of our estimator by changing one sample in $x^n$. The worst case change in the value of our estimator is when we take a symbol in $S$ (or $\bar{S} \backslash \{y'\}$) and replace it by a symbol in $\bar{S}\backslash \{y'\}$ (or $S$). In this case, by triangle inequality change in our estimator is upper bounded by $8\max\left(e^{L^2/n}\max|b_i|, \frac{L_g}{n}, g\left(\frac{c_1 \log (n)}n\right), g_n\right) + G_{L,g}(1)$ that is further upper bounded by $9\max\left(e^{L^2/n}\max|b_i|, \frac{L_g}{n}, g\left(\frac{c_1 \log (n)}n\right), g_n\right)$ and the result follows.

	\end{proof}

\subsection{Entropy}\label{subsec:entr}
Here we present proof sketch for the following: for entropy the estimator defined above has low bias and the value of the estimator does not change too much by change in one sample. This result is analogous to Lemma 6 in~\cite{ADOS16} and our proof for this lemma is very similar to~\cite{ADOS16} and for completeness sketch for the proof.
\begin{lemma}
  Let $g_n = 1/(2n)$ and $\alpha > 0$.  Suppose $c_1=2c_2$, and $c_2>35$,
  Further suppose that
  $n^3\left(\frac{16c_1}{\alpha^2}+\frac1{c_2}\right)>\log k\cdot\log
  n$. Then for all subset $S\subseteq \bX$, there exists a polynomial approximation of $- y \log y$ {with degree $L = 0.25 \alpha \log n$}, over
  $[0, c_1\frac{\log N}{n}]$ such that $\max_{i} |b_i| \leq
  n^{\alpha}/n$ and the bias of the entropy estimator on subset $S$ ($\sum_{y \in S}\bp_{y} \log \frac{1}{\bp_{y}}$) is at most
  $\cO\left(\left(1+\frac{1}{\alpha^2}\right)\frac{\absz}{n\log N} + \frac{\log N}{N^4}\right)$.
\end{lemma}
\begin{proof}
	\newcommand{\eo}{{E}_{1}}
	\newcommand{\et}{{E}_{2}}
	\newcommand{\ethr}{{E}_{3}}
	\newcommand{\eall}{{E}}
	\newcommand{\ny}{n_{y}}
	\newcommand{\nyp}{n'_{y}}
	\newcommand{\expt}[1]{\mathbb{E}\left[ #1 \right]}
	
We first upper bound the bias of our estimator. We consider three events, 
$$\eo\defeq \cap _{y \in S} \left\{ \nyp \leq  c_{2} \log N, \nyp \leq  c_{1} \log N \implies \bp_{y} \leq \frac{c_1 \log N}{n}\right\},$$ 
$$\et\defeq\cap _{y \in S} \left\{ \nyp > c_{2} \log N \implies \bp_{y} > \frac{c_3 \log N}{n}\right\},$$ $$\ethr\defeq \cap _{y \in S}\left\{ \nyp \leq  c_{2} \log N \text{ and } \nyp > c_{1} \log N \right\}~.$$ 
By proof of Lemma 6 in~\cite{ADOS16} we have,
\begin{equation}\label{prob:e3}
\Prob{\ethr^c} \leq \frac{1}{n^{4.9}}
\end{equation}
By equations 48 and 49 in~\cite{WY16} combined with \Cref{prob:e3}, we get,
$$\Prob{\eo^c} \leq \frac{2}{N^4} \text{ and } \Prob{\et^c} \leq \frac{1}{N^4} $$
Define $\eall\defeq \eo \cap \et$, then 
\begin{equation}\label{prob:eall}
\Prob{\eall^c} \leq \Prob{\eo^c}+\Prob{\et^c} \leq \frac{2}{N^4}~. 
\end{equation}

\newcommand{\io}{I_1}
\newcommand{\ito}{I_2}
Further we define random sets
$\io\defeq \left\{y \in S |  \nyp < c_{2} \log N, \ny < c_{1} \log N  \text{ and } \bp_{y} \leq \frac{c_1 \log N}{n}\right\}$ and $\ito \defeq \left\{ y \in S ~|~\nyp > c_{2} \log N \text{ and } \bp_{y} > \frac{c_3 \log N}{n}\right\}$.
We first bound the conditional bias and we later use it to bound the bias of our estimator. Our next statement follows from uniform approximation error~\cite{Tim63} and is explicitly written in to Equation 53 of~\cite{WY16}.
\begin{equation}\label{eq:2}
|\expt{\bff_{\io}(\bp)-\hat{g}_{\io} | \io}=|\sum_{y \in \io} \bp_{y} \log \frac{1}{\bp_y}-P_{L,g}(\bp_{y})| \leq \frac{N}{\alpha^2 n \log N}~.
\end{equation}
Let $\hat{\bp}$ be the empirical distribution on $\xtn$. Similarly by analysis of Case 2 and Equation 58 in~\cite{WY16} we have,
\begin{equation}\label{eq:3}
|\expt{\bff_{\ito}(\bp)-\hat{g}_{\ito} | \ito}=|\expt{\sum_{y \in \ito} (\bp_{y} \log \frac{1}{\bp_y}-\hat{\bp}_{y} \log \frac{1}{\hat{\bp}_y})|\ito}| \leq \frac{N}{n \log N}~.
\end{equation}
Combining equations \ref{eq:2}, \ref{eq:3}, \ref{prob:eall} and \ref{prob:e3} we can upper bound the bias of our estimator by $\frac{2N}{n \log N}+\frac{4 \log N}{N^4}$. Note here we use the fact that in the case of bad event ($\eall^c$ or $\ethr^c$) the bias of our estimator is upper bounded by $\log N$.

Our analysis for largest change in the value of estimator by changing one sample is exactly the same as~\cite{ADOS16} and for completeness we describe it next. The largest coefficient of the optimal uniform polynomial approximation of degree $L$ for function $x \log x$ in the interval $[0,1]$ is upper bounded by $2^{3L}$. This result follows from the proof of Lemma 2 in~\cite{CL11} and is also explicitly mentioned in the proof of Proposition 4 in~\cite{WY16}. Therefore, the largest change (after appropriately normalizing) is the largest value of $b_i$ (co-efficient of the optimal uniform polynomial approximation) which is
\[
\frac{2^{3L}e^{L^2/n}}{n}.
\]
For $L = 0.25\alpha\log n$, this is at most $\frac {n^\alpha}{n}$. 
\end{proof}

The proof of Lemma~\ref{thm:entrochange} for entropy follows from the above lemma and Lemma~\ref{lem:general_est} by substituting $n= \cO\left(\frac{N}{\log N} \frac{1}{\dst}\right)$ and $\epsilon>\Omega\left(\frac{\log N}{N^{1-\alpha}}\right)$.

\subsection{Distance to uniformity}\label{subsec:dtu}
Here we provide proof sketch for the existence of an estimator with desired properties. This result is analogous to Lemma 7 in~\cite{ADOS16} and proof for this lemma is very similar to that of \cite{ADOS16} and for completeness we sketch the proof for this result. 
\begin{lemma}

  Let $c_1>2c_2$, $c_2=35$. Then for all subset $S\subseteq \bX$, there is an estimator for distance to uniformity on subset $S$ ($\sum_{y \in S}|\bp_{y}-\frac{1}{N}|$) that changes by at most $n^{\alpha}/\nsmp$ when a sample is changed, and the bias of the estimator is at most $O(\frac{1}{\alpha}\sqrt{\frac{c_1\log N}{N\cdot n}})$.
\end{lemma}
\begin{proof}
We divide estimation of distance to uniformity into two cases based on $n$. Note the proof for this lemma follows along the lines of \cite{ADOS16}.
\paragraph{Case 1: $\frac1\absz < c_2\log N/n$.} In this case, we use the estimator defined in the last section for $g(x) = |x-1/k|$.

\paragraph{Case 2: $\frac1\absz > c_2\log N/n$.} The estimator is as follows for all $y \in S$:
\[
g_{\smb}=
\begin{cases}
G_{L,g}(\Mltsmb), & \text{ for }
\absv{\frac{\Mltsmb^{'}}{n}-\frac1{\absz}}<\sqrt{\frac{c_2\log N}{Nn}}, \text{
  and } \absv{\frac{\Mltsmb}{n}-\frac1{\absz}}<\sqrt{\frac{c_1\log N}{Nn}},\\ 0,
& \text{ for }
\absv{\frac{\Mltsmb^{'}}{n}-\frac1{\absz}}<\sqrt{\frac{c_2\log N}{Nn}}, \text{
  and } \absv{\frac{\Mltsmb}{n}-\frac1{\absz}}\ge\sqrt{\frac{c_1\log N}{Nn}},\\ g
\left(\frac{\Mltsmb}{n}\right) , & \text{ for }
\absv{\frac{\Mltsmb^{'}}{n}-\frac1{\absz}}\ge\sqrt{\frac{c_2\log N}{Nn}}.
\end{cases}
\]

The estimator proposed in~\cite{ADOS16} is exactly the same as ours, but we define our estimator only for the domain elements in $S\subseteq \bX$. Note the estimator defined in~\cite{JHW17} is slightly different, assigning $G_{L,g}(\Mltsmb)$ for the first two cases. As in~\cite{ADOS16}, this second case is designed to bound the change in value of the estimator by changing one sample. Using~\cite{Tim63}(Equation 7.2.2),~\cite{JHW17} (for their estimator) show that, contribution towards bias (conditioned on "good" event \footnote{\label{event}Refer~\cite{JHW17} for definitions of "good" and "bad" events.}) by any domain element $y \in \bX$ (note we only need this result to hold for $y\in S$) satisfying $\Mltsmb^{'}<c_2\log N, \text{ and } \Mltsmb< c_1\log N$ for case 1 and $\absv{\frac{\Mltsmb^{'}}{n}-\frac1{\absz}}<\sqrt{\frac{c_2\log N}{Nn}}, \text{ and } \absv{\frac{\Mltsmb}{n}-\frac1{\absz}}<\sqrt{\frac{c_1\log N}{Nn}}$ for case 2, using polynomial approximation (Lemma 27 in~\cite{JHW17}) is upper bounded by $\cO\left(\frac1{L}\sqrt{\frac{\log N}{N\cdot n\log N}}\right)$, where $L$ is the degree of optimal uniform approximation for function $|x-\frac{1}{N}|$ in the interval $[0,2c_1 \frac{\log N}{N}]$ for case 1 (Equation (351) in~\cite{JHW17}) and $[\frac{1}{N}-\sqrt{\frac{c_1\log N}{Nn}},\frac{1}{N}-\sqrt{\frac{c_1\log N}{Nn}}]$ for case 2 (Equation (367) in~\cite{JHW17}). Further, the bias (conditioned on the "good" event) by empirical estimate for domain element $y \in \bX$ (as before, we need this result to hold only for $y\in S$) satisfying $\Mltsmb^{'}\ge c_2\log N$ for case 1 and $\absv{\frac{\Mltsmb^{'}}{n}-\frac1{\absz}}\ge\sqrt{\frac{c_2\log N}{Nn}}$ for case 2, is zero (Refer proof of Theorem 2~\cite{JHW17}).~\cite{JHW17} also bound the probability of "bad" event\footnoteref{event} (Refer proof of Lemma 2 in~\cite{JHW17}), thus bounding the bias with respect to these domain elements. Further similar to~\cite{ADOS16}, by our choice of $c_1, c_2$, the contribution to bias by domain element $y \in S$ satisfying $\Mltsmb^{'}<c_2\log k, \text{
	and } \Mltsmb\ge c_1\log N$ for case 1 and $\absv{\frac{\Mltsmb^{'}}{n}-\frac1{\absz}}<\sqrt{\frac{c_2\log N}{Nn}}, \text{
	and } \absv{\frac{\Mltsmb}{n}-\frac1{\absz}}\ge\sqrt{\frac{c_1\log N}{Nn}}$ for case 2, is upper bounded by $1/n^4<\dst^2$. Combining analysis of all these cases together, we have our result for bias.

The proof for largest change in the estimator value by changing one sample is exactly same as~\cite{ADOS16}. Similar to~\cite{ADOS16}, here we use the fact~\cite{CL11} (Lemma 2) that the largest coefficient of the optimal uniform polynomial approximation of degree $L$ for function $|x|$ in the interval $[-1,1]$ is upper bounded by $2^{3L}$. $2^{3L}$. Similar to entropy (after appropriate normalization), the largest difference in estimation will be at most $n^{\alpha}/n$.
\end{proof}

The proof of \Cref{thm:dtuchange} for distance to uniformity follows from the above lemma and \Cref{lem:general_est} and by substituting $n= \cO\left(\frac{N}{\log N} \frac{1}{\dst^2}\right)$.

\end{document}